\documentclass[11pt,letterpaper]{article}
\usepackage{setspace}
\usepackage[margin=1in, paper=letterpaper]{geometry}

\usepackage{microtype}
\usepackage{amsmath,amssymb,amsfonts}
\usepackage{amsthm}
\usepackage[utf8]{inputenc}
\usepackage{hyperref}
\usepackage[page]{appendix}
\usepackage{tikz}

\usepackage{xspace}
\usepackage{delimset}
\usepackage{xcolor}
\usepackage{thmtools}
\usepackage{cleveref}

\newtheorem{theorem}{Theorem}
\newtheorem{lemma}{Lemma}

\newtheorem{remark}{Remark}
\newtheorem{observation}{Observation}
\Crefname{observation}{Observation}{Observations}

\newcommand{\size}{\delim{|}{|}}
\newcommand{\ceil}{\delim{\lceil}{\rceil}}

\newcommand{\set}{\delimpair{\{}{[m]|}{\}}}
\newcommand{\brc}{\brk[c]}

\let\eps\epsilon

\newcommand*{\PRED}{\textup{\tt pred}\xspace}

\newcommand*{\RANK}{\textup{\tt rank}\xspace}
\newcommand*{\SELECT}{\texttt{\tt select}\xspace}

\newcommand*{\SUM}{\texttt{\tt sum}\xspace}
\newcommand*{\SEARCH}{\texttt{\tt search}\xspace}
\newcommand*{\RANKM}{\texttt{\tt rankm}\xspace}

\newcommand*{\MSB}{\texttt{\tt msb}\xspace}

\DeclareMathOperator{\polylog}{polylog}
\DeclareMathOperator{\poly}{poly}
\renewcommand*{\O}{\mathcal{O}}

\usepackage[inline]{enumitem}
\setlist[enumerate]{label=(\arabic*)}

\usepackage{natbib}
\def\Clambda{1024}
\def\Cprobe{64}
\def\Cfactor{64}
\def\Cgamma{16}
\def\Cdelta{8}

\renewcommand{\Pr}{\operatorname{Pr}\brk[s] }
\newcommand*{\Prc}{\operatorname{Pr}\delimpair {[} {[m]|} {]} }
\DeclareMathOperator{\dist}{dist}

\title{Achieving Optimal Redundancy for Small Dynamic Rank/Select Dictionaries}

\begin{document}

\author{Gabriel Marques Domingues\footnote{
{\tt {gm@mail.tau.ac.il}}.
School of Electrical Engineering, Tel Aviv University.
This research was supported by the Israel Science Foundation, grant No. 1948/21.}}
\date{}

\begin{titlepage}
\maketitle

\begin{abstract}
In this paper, we study the number of bits required to construct a dynamic dictionary with optimal time for $\texttt{rank}/\texttt{select}$ operations.
Using the standard (multiplication) Word-RAM model
with $w$-bit words,
we construct a data-structure for a dynamic $\texttt{rank}/\texttt{select}$ dictionary for
a set $S\subseteq\{0,1,\cdots,u-1\}$ of $n$ elements that,
given a parameter $1\leq k\leq \log^*w$,
uses 
$$\operatorname{lg}\binom{u}{n}+\mathcal{O}(n\log^{(k)}w)\text{ bits}$$
taking optimal $\mathcal{O}(k+\log_w n)$ time (worst-case) for all operations.
We show optimality for $n=w^{\mathcal{O}(1)}$ and $\log u=\Theta(w)$
by extending the lower bound of 
Li, Liang, Yu, and Zhou [FOCS 2023]
to super-polynomial universes:
any dynamic dictionary for $n\leq \sqrt{u}$ elements
and $\log u=\Theta(w)$
that uses $\operatorname{lg}\binom{u}{n}+\mathcal{O}(n\log^{(k)}n)$ bits 
requires $\Omega(k)$ time for operations.
Lastly, we extend the data-structure to a dynamic fully indexable dictionary (that also supports $\texttt{rank}/\texttt{select}$ on the complement of $S$).
\end{abstract}
    
\end{titlepage}

\section{Introduction}

A classic data-structure problem is to design
dynamic $\RANK/\SELECT$ dictionaries using as few bits as possible.
There has been a lot of recent development in these constructions for
the Word-RAM model~\cite{pibiri2020succinct,li2023dynamic,kuszmaul2026succinct,domingues2026compressing}.
A $\RANK$/$\SELECT$ dictionary is a data-structure that encodes a set $S\subseteq\{0,1,\cdots, u-1\}$ and supports the following operations:
\begin{enumerate}
\item $\RANK(S, x)$ returns the number of elements of $S$ less than or equal to $x$, and
\item $\SELECT(S, i)$, for $1\leq i\leq \size{S}$, returns the $i$-th smallest element in $S$ .
\end{enumerate}
These can be combined to support other operations,
namely, predecessor $\PRED(S, x)$, which returns the largest element of $S$ less than or equal to $x$, and membership $(x\in S)$, which returns whether $x$ belongs to $S$.
Throughout the paper, we focus on the dynamic case, which supports insertions and deletions to and from $S$.

For a dynamic $\RANK/\SELECT$ dictionary,~\cite{fredman1989cell} proved that the optimal time to implement all operations in Word-RAM with $w$-bit words is $\Omega(\log_w n)$, where $n=\size{S}$.
The information-theoretic lower bound to encode the dictionary is $\ceil{\log_2\binom{u}{n}}$.
We say a data-structure uses $R$ redundancy if it uses $\ceil{\log_2\binom{u}{n}}+R$ bits.
Most existing constructions of dynamic $\RANK$/$\SELECT$ dictionaries with optimal time for operations have redundancy $\Omega(n)$~\cite{pibiri2020succinct,domingues2026compressing}.
Using a randomized data-structure,~\cite{kuszmaul2026succinct} achieved $\O(nw/2^{(\log_w n)^{1/3}})$ redundancy, which is $o(n)$ for $n\geq 2^{\omega(\log^4 w)}$, with optimal amortized time 
$\O(\log_w n)$ in expectation.
For deterministic data-structures,
the current best result is that of~\cite{domingues2026compressing},
achieving $\O(nw^\eps)$ redundancy for any constant $\eps\leq1/2$.
This construction uses the standard Word-RAM model, that is, only multiplication and standard $w$-bit operations (addition, bit-shifts, etc.) are constant time.
In the cell-probe model, where \emph{all} $w$-bit operations are constant-time,
the construction of~\cite{domingues2026compressing} can be improved to use $\O(n\log w)$ redundancy using the data-structure of~\cite{ajtai1984hash}.

As mentioned by~\cite{kuszmaul2026succinct},
there are no known lower bounds for the redundancy in the dynamic case.
This shows a significant
gap between the dynamic case
and the static case, where a lower bound of $\Omega(n/\polylog u)$ redundancy for $\O(1)$ time is known and reached~\cite{patrascu2010cell,liang2025optimal}.
A related lower bound was proven by~\cite{li2023tight}
for dynamic dictionaries (i.e., supporting only membership, insertions, and deletions),
yet it holds only for polynomial universes with $w=\Theta(\log u)$.
Denote by $\log^{(k)}(n)$ the (binary) logarithm applied $k$ times to $n$, and let $\log^*(n)=\min\set{k}{1\leq \log^{(k)}(n)<2}$.
The lower bound states that, given $k\leq \log^*n$,
a dictionary using $\O(n\log^{(k)}n)$ redundancy requires $\Omega(k)$ time (amortized).
However, for $\RANK$/$\SELECT$ dictionaries for polynomial universes,
the lower bound of~\cite{fredman1989cell} is stronger,
noting that $\Omega(\log_w n)=\Omega(\log_w u)=\Omega(w/\log w)=\Omega(k)$.
A recent mechanized proof of~\cite{dambrosio2026always}
extends the proof structure for super-polynomial universes, focusing on the case of $R<n$.

In this paper, 
we provide matching upper and lower bounds for $\RANK$/$\SELECT$ dictionaries
in the standard Word-RAM model for small sets $S$.
First, we advance the construction of~\cite{domingues2026compressing}
to achieve $\O(n\log^{(k)} n)$ redundancy
and $\O(k)$ time for all operations 
for a set of $n=w^{\O(1)}$ elements.
In addition, we extend it to implement $\RANK/\SELECT$ on the complement of $S$, achieving a so-called \emph{fully indexable dictionary} (FID).
Using the techniques of~\cite{domingues2026compressing},
we generalize this construction to any set of $n$ elements taking $\O(k+\log_w n)$ time
and using redundancy $\O(n\log^{(k)} w)$.

Second, we adapt the techniques of~\cite{li2023tight} to show that
any dynamic dictionary for $n\leq \sqrt{u}$ elements using
$\O(n\log^{(k)}n)$ redundancy requires $\Omega(k)$ time.
In fact, this lower bound holds even for dynamic dictionaries
that support only membership queries.
The main difficulty of the lower bound of~\cite{li2023tight}
is to remove the requirement of a polynomial universe.
We focus on the case of $R=\Omega(n)$ redundancy, and simplify the proof for this case,
avoiding the technical issue of~\cite{li2023tight} due to conditioning (see~\Cref{appendix:inner}),
which was also resolved in~\cite{dambrosio2026always}.

For $n=w^{\O(1)}$ elements,
our upper bound is optimal
and shows a fundamental limitation of compressing dynamic $\RANK$/$\SELECT$ dictionaries in Word-RAM.
This closes the problem of constructing constant-time fusion nodes as constant-time is not possible for $o(n)$ redundancy~\cite{li2023tight,dambrosio2026always}.
Moreover, this result shows that the standard Word-RAM model
(i.e., using only multiplication, addition, bit-shifts, bitwise operations, reads and writes)
suffices to solve this problem.

For $k=\log^*n$,
this is the first deterministic construction
using $\O(n)$ redundancy
that achieves $\O(\log^*w)=\log^{o(1)} w$ time for sets of $\poly(w)$ elements.
With this lower bound,
the remaining open question is how to achieve
$o(n)$ redundancy with the optimal time-space trade-off,
i.e., taking time $\Omega(\log^* n+\log(n/R))$ for redundancy $R$~\cite{li2023tight,dambrosio2026always}.
Currently, the cardinality reduction techniques of~\cite{domingues2026compressing} require $\Omega(n)$ bits to get worst-case update times.

\section{Contributions and Techniques}

In this section,
we describe our results and our techniques.
We employ the standard Word-RAM model with $w$-bit words.
Recall that a dynamic $\RANK$/$\SELECT$ dictionary encodes a set $S$ of $n$ integers in $[u]=\brc{0,1, \cdots,u-1}$ supporting $\RANK/\SELECT$ and updates.
A dictionary is said to have redundancy $R$ if it uses $\lg\binom{u}{n}+R$ bits, 
where we denote $\lg(a)=\ceil{\log_2(a)}$.
Throughout, we employ the following notation: $$\log^{(k)}(a)=\overbrace{\log\cdots\log}^{k}(a)~~\text{and}~\log^*(a)=\min\set{k}{1\leq\log^{(k)}(a)<2}~.$$

\subsection{Dynamic Rank/Select Dictionaries}

For our upper bounds,
all running times are worst-case and all data-structures are deterministic.
The memory model for dynamic memory allocation is $\mathcal{M}_B$~\cite{raman2003succinct},
in which the space utilization is defined to be the highest memory cell in use.
We build the following data-structure for a dynamic $\RANK$/$\SELECT$ dictionary of $n=w^{\O(1)}$ elements.

\begin{restatable}{theorem}{ThmDict}\label{thm:dictionary}
For every $1\leq k\leq \log^*n$,
there is a dynamic $\RANK$/$\SELECT$ dictionary
for $n=w^{\O(1)}$ elements
using $\O(n\log^{(k)}n)$ redundancy
and taking $\O(k)$ time for $\RANK/\SELECT$ and update.
\end{restatable}

For $k=1$, this has already been achieved in~\cite{domingues2026compressing},
using the cell-probe $\RANK$ data-structure of~\cite{ajtai1984hash},
but only $\O(nw^\eps)$ in the standard Word-RAM (for constant $\eps\leq 1/2$).
Moreover,~\cite{domingues2026compressing} views the dictionary as two separate data-structures: one constant-time data-structure for $\SELECT$ (essentially the sorted sequence of elements), and one for $\RANK$ that can access the $\SELECT$ data-structure (called a $\RANK$ index).
We will improve this construction for $\SELECT$ using a recursive construction of depth $k$ on the permutation of the elements in the sequence.
Some care is required to manage memory in $\mathcal{M}_B$ using only $\O(n\log^{(k)}n)$ additional bits, within $\O(k)$ time.
We employ
the Elias-Fano encoding~\cite{elias1974efficient,fano1971number,carter1978exact,pibiri2017dynamic,pibiri2020succinct}
to reduce the number of bits for the elements stored in the sequence.

To implement $\RANK$, we perform a non-trivial modification to this data-structure for a universe $\log u=\O(\sqrt{w})$,
that results in a dynamic sequence 
using only $\O(n)$ redundancy,
and that further allows us to access a contiguous set of $\Theta(w^{1/4})$ elements in $\O(1)$ time.
Consequently, we can compute $\RANK$ in this set in $\O(1)$ time.
For a general universe, we take the prefix of $\ell=\Theta(\sqrt{w})$ bits of the elements and store it separately using this idea.
The main trick is the following observation:
for a set of $\Theta(w^{1/4})$ elements,
either all prefixes are distinct,
in which case we compute $\RANK$ only by the prefixes,
or there is a repetition,
in which case we reuse the bits in one of the repeated prefixes
to store a $\RANK$ index.
This adds support for $\RANK$ and only 
incurs $\O(n)$ redundancy, $\O(1)$ time, and $\O(1)$ accesses to $\SELECT$.

Lastly,
we extend these results to a fully indexable dictionary (FID) using similar ideas.
A FID is a dictionary that supports $\RANK$ and $\SELECT$ for both $S$ and its complement $\overline{S}=[u]\setminus S$.
We claim the following theorem.

\begin{restatable}{theorem}{ThmFID}\label{thm:FID}
For every $1\leq k\leq \log^*n$,
there is a dynamic FID
for $n=w^{\O(1)}$ elements
using $\O(n\log^{(k)}n)$ redundancy
and taking $\O(k)$ time for all operations.
\end{restatable}

Using the cardinality reduction lemma from~\citep[{\S3}]{domingues2026compressing} to the previous result for sets of $\Theta(w^2)$ elements, we extend these results to arbitrary $n$.

\begin{theorem}\label{thm:general}
For every $1\leq k\leq \log^*w$,
there is a dynamic FID
for $n$ elements
using $\O(n\log^{(k)}w)$ redundancy
and taking $\O(\log_wn + k)$ time for all operations.
\end{theorem}

For $k=\O(1)$,
this is the first deterministic dynamic FID with $o(n\log w)$ bits with time $\O(\log_wn)$ for all operations.
For $k=\log^*w$, this achieves $\O(n)$ redundancy.
This is unlikely to be improved upon using our techniques,
as we employ the Elias-Fano encoding,
which always requires $\Omega(n)$ bits of redundancy (cf.~\cite{pibiri2017dynamic,pibiri2020succinct}).
Moreover, the de-amortization of the cardinality reduction of~\cite{domingues2026compressing} incurs an additional $\O(n)$ redundancy.

\subsection{Lower Bound}

We prove the following theorem.
Here, a \emph{dictionary} only supports membership and updates.
Our proof is a modification of the techniques of~\cite{li2023tight} to handle the case $u\geq n^2$.
In fact, the relevant case is where $n=w^{\O(1)}$ and $\log u=\Theta(w)$.

\begin{theorem}[Slight generalization of~\cite{li2023tight}]\label{thm:lower-bound}
For $1\leq k\leq\log^*n$
and $\log u=\Theta(w)$,
a dictionary storing $n\leq \sqrt{u}$ elements
using $R=\O(n\log^{(k)}n)$ redundancy
must take at least $\Omega(k)$ amortized time in expectation
to answer membership queries or updates.
\end{theorem}

This result,
together with the $\Omega(\log_w n)$ lower bound of~\cite{fredman1989cell},
shows that our construction in~\Cref{thm:dictionary} is optimal for the regime $n=w^{\O(1)}$.
The remaining challenge is to consider the regimes that can achieve $o(n)$ redundancy.
The case $R<n$ is handled by~\cite{dambrosio2026always},
and shows a lower bound of $\Omega(\log(n/R))$ time, using a more involved statement.

The proof of~\cite{li2023tight} considers the following hard instance for a sequence of operations.
Let $K$ and $K'$ be two uniformly random disjoint subsets of $[u]$.
Let $D$ and $E$ be uniformly random permutations of $K$ and $K'$, respectively.
The hard instance inserts all elements of $K$ into the dictionary, then, for every $i$ in $[n]$, the instance queries $D[i]$, deletes $D[i]$, and inserts $E[i]$.
We refer to the sequence of query, deletion, and insertion, as a \emph{meta-operation}.

Our goal is to show the following lemma, which implies~\Cref{thm:lower-bound}.
By Yao's minimax principle~\cite{yao1977probabilistic}, we can assume that the dictionary is deterministic.

\begin{restatable}[Slight generalization of~\cite{li2023tight}]{lemma}{LemmaProbe}\label{lemma:probe}
For $1\leq k\leq\log^*n$
and $\log u=\Theta(w)$,
a dictionary storing $n\leq \sqrt{u}$ elements
using $R=\O(n\log^{(k)}n)$ redundancy
must take $\Omega(nk)$ cell probes in expectation for all meta-operations in the hard instance.
\end{restatable}

\section{Preliminaries}

\paragraph*{Word-RAM model.}
We employ the Word-RAM model, where the memory is an array of $2^w$ words of $w$-bits~\cite{yao1981should,fredman1989cell};
memory reads and writes of a single word take $\O(1)$ time.
We restrict all $w$-bit operations to standard operations (addition, bitwise operations, shifts), and multiplication~\cite{patrascu2014dynamic,domingues2026compressing}.
We allow using lookup tables of $\polylog w$ words
(i.e., $\O(w\polylog w)$ bits),
as long as they are computable in linear time.
Many operations, such as $\MSB$ can be implemented using a combination of lookup tables and multiplication.

\paragraph*{Memory model.}
For dynamic memory management,
we use the $\mathcal{M}_B$ model of~\cite{raman2003succinct},
which defines the space used by a data-structure
to be the largest address in the memory array of words
that it uses (we measure it in bits).
Namely,
the last used cell is what matters for the memory usage,
even if we do not use some intermediate cells.
The main difficulty in this model is to concatenate two data-structures~\cite{li2023dynamic,li2024dynamic}.
As~\cite{domingues2026compressing} points out, it is possible to concatenate two structures that take $s_1$ and $s_2$ bits, respectively, using $s_1+\O(s_2)$ bits.

\paragraph*{Lazy updates.}
For many of our data-structures,
it is easier to describe updates in an amortized fashion.
For those cases, there are standard techniques to de-amortize these by performing the updates piecemeal~\cite{willard2000examining,andersson2007dynamic,domingues2026compressing}.
One recurring instance of this is when splitting a sequence $A$ into segments of $\Theta(s)$ elements, for a given parameter $s$.
If the size of a segment halves/doubles, we need to merge/split them appropriately.
To do this lazily, at every access/update to these segments, we move $\O(1)$ elements from one segment to another.
We choose the constant large enough so that this procedure finishes before any of the segments want to merge/split again.
See~\cite{domingues2026compressing} for details.

\subsection{Previous constructions}

We give a short overview of the techniques described in~\cite{domingues2026compressing} that will be relevant for this construction.

\paragraph*{Rank indices.}
Previous constructions~\cite{patrascu2014dynamic,domingues2026compressing}
decomposed the problem into two data-structures:
\begin{enumerate}[label=(\Roman*)]
\item The $\SELECT$ dictionary stores the set $S$ and supports $\SELECT$ and update.
\item The $\RANK$ index accesses the $\SELECT$ dictionary and supports $\RANK$ and update. 
\end{enumerate}
We cite the compressed $\RANK$ index of~\cite{domingues2026compressing}, which we will use in our constructions as a black-box.

\begin{theorem}[Compressed $\RANK$ index~\cite{domingues2026compressing}]\label{thm:domingues}
For every constant $0<\eps\leq 1/2$, there is a $\RANK$ index for $n=w^{\O(1)}$ elements using $\O(nw^\eps)$ bits and constant-time $\RANK$ and update, using $\O(1)$ accesses to the $\SELECT$ dictionary.
\end{theorem}

\paragraph*{Partial-sums.}
One main data-structure we will use many times is
the searchable partial-sums~\cite{fredman1982complexity,dietz1989optimal,patrascu2004lower,patrascu2006logarithmic,arroyuelo2016succinct,bille2018dynamic,domingues2026compressing}.
The standard presentation of this problem is to encode a sequence $A$ supporting $\SUM(A,i)=\sum_{j=0}^i A[j]$ and $\SEARCH(A,z)=\max\set{i}{z \geq \SUM(A,i)}$,
and updates of the form $A[i]\gets A[i]\pm 1$.
Note that $\RANK/\SELECT$ is a special case of this, for a bit-vector.
We prefer the following description:
encode a (non-decreasing) sorted array $Z[0:n-1]$ supporting access $Z[i]$,
updates of the form $Z[i:n-1]\gets Z[i:n-1]\pm 1$, and
$$\RANKM(Z,z)=\max\set{0\leq i\leq n-1}{z \geq Z[i]}\text{ ($-1$ if the set is empty)}~.$$
One can think of this as $\RANK/\SELECT$ on a multiset (a set with repetitions).
We also support inserting and deleting entries of $Z$.
For small values, there is a simple construction in the standard Word-RAM model.
\begin{lemma}[{\citep[{\S3}]{domingues2026compressing}}]\label{lemma:partial-sum-small}
For an array $Z$ of $n=w^{\O(1)}$ elements with $\lg\sigma=\O(\log w)$ bits,
there is a partial-sum data-structure using $\O(n\lg\sigma+n\log n)=\O(n\log w)$ bits and taking $\O(1)$ time for all operations.
\end{lemma}
\noindent
For larger alphabets, we have to use a $\RANK$ index. We cite the following construction.
One can think of this as ``partial-sum search index'', analogous to the $\RANK$ index.
\begin{lemma}[{\citep[{\S3}]{domingues2026compressing}}]\label{lemma:partial-sum}
Given $\O(1)$ time access to the array $Z$ of $n=w^{\O(1)}$ elements,
given a constant-time $\O(nr)$-bit $\RANK$ index,
there is a partial-sum data-structure with $\O(n\log w+nr)$ additional bits and $\O(1)$ time for all operations.
\end{lemma}

\paragraph*{Cardinality reduction.}
For a $\RANK/\SELECT$ dictionary, and in fact, a fully indexable dictionary, we cite the cardinality reduction of~\cite{domingues2026compressing}, that reduces the problem of encoding an arbitrary set to encoding only $\poly(w)$ elements.
Combining this with~\Cref{thm:FID}, then~\Cref{thm:general} follows.

\begin{theorem}[Cardinality reduction~\cite{domingues2026compressing}]
Given a dynamic FID for $n=\O(w^2)$ elements using $\O(nr)$ redundancy and taking $\O(t)$ time for operations (worst-case),
there is a dynamic FID for an arbitrary $n$ number of elements using $\O(nr+n)$ redundancy and taking $\O(t+\log_w n)$ time for operations (worst-case).
\end{theorem}

\noindent
See~\cite{domingues2026compressing} for details on how to reduce $\RANK$, $\SELECT$, and partial-sums to a small set of size $\Theta(s)$, for a given parameter $s$, by using a $\RANK$ index and partial-sum structure for $\Theta(n/s)$ elements.

\section{Upper bound}

In this section,
we describe upper bounds for dynamic $\RANK/\SELECT$ dictionaries and extend it to fully indexable dictionaries.
In the next section, we will show that our results are tight.

\subsection{Achieving low redundancy}\label{sec:sequence}

We describe a construction for a $\SELECT$ dictionary for $n=w^{\O(1)}$ elements that matches the lower bound.
We implement the $\SELECT$ dictionary as a \emph{dynamic sequence}, that is, a data-structure that encodes a sequence $A$ while supporting insertion, deletion, and access.
For the $\SELECT$ dictionary, we use $A[i]=\SELECT(S,i)$.
The main issue here is to carefully manage the memory using the $\mathcal{M}_B$ model.

We recall the construction of~\cite{domingues2026compressing}, 
that achieved $\O(n\log n)$ redundancy.
The idea there is to store an unordered array of elements contiguously and maintain the permutation of elements separately, using $\O(n\log n)$ bits.
Let $A$ be the sequence that we want to represent.
Let $A^*$ be an array of $\lg u$ bits which stores the elements explicitly and let $P$ be the permutation of $A$ in $A^*$, that is, $A^*[P[i]]=A[i]$.
For $n=w^{\O(1)}$, the proof of~\Cref{lemma:previous} shows that we can encode $P$ in $\O(n\log n)$ with constant-time operations using an array of $(\log n)$-bit indices.
Every time $\log n$ increases/decreases by a sufficiently large constant,
we grow/shrink the indices;
this can be de-amortized using lazy updates.

\begin{lemma}[\cite{domingues2026compressing}]\label{lemma:previous}
There is a data-structure encoding a dynamic sequence
of $n=w^{\O(1)}$ elements in $[u]$
using an array of $n\lg u$ bits plus
$\O(n\log n)$ bits for indexing
and taking $\O(1)$ time for access, insert, and delete.
\end{lemma}

Our construction will improve this result using the more involved layout in the array.
The main issue is how to avoid having an arbitrary permutation $P$ of the elements, as that requires $\Omega(n\log n)$ bits.
Instead, we organize the elements in $A^*$ in a more structured way
to ``restrict'' the permutation.

\begin{theorem}\label{thm:sequence}
For every $1\leq k\leq \log^* n$,
there is a data-structure encoding a dynamic sequence
of $n=w^{\O(1)}$ elements in $[u]$
using an array of $n\lg u$ bits plus
$\O(n\log^{(k)} n)$ bits for indexing
and taking $\O(k)$ time for access, insert, and delete.
\end{theorem}
\begin{proof}
As in~\Cref{lemma:previous}, we will use an out-of-order array $A^*$ to keep the elements and have additional data-structures to index into $A^*$ and handle updates (analogous to the role of $P$ in~\Cref{lemma:previous}).
In particular, the  external array $A^*$ grows/shrinks at the end
with $\O(1)$ time reads and writes.
The theorem holds even if the array $A^*$ is not stored contiguously in memory,
as we will use later.

We divide the sequence $A$ into 
``superblocks'' of $\Theta(\log^4 n)$ successive elements.
Using a partial-sum data-structure of~\Cref{lemma:partial-sum-small},
we can find the superblock to which an index belongs in $\O(1)$ time, at the cost of $\O((n/\log^4n)\cdot \log n)=\O(n/\log^3 n)$ additional bits.
In addition, each superblock is divided into $\Theta(\log^2 n)$ unsorted blocks of $\log^2 n$ elements (except for the very last block in a superblock, which can have less than $\log^2 n$ elements).
Successive elements may be stored in different blocks.
We call the last $\Theta(1)$ blocks of a superblock the ``buffers''.
In our memory organization, we start with the non-buffer blocks,
then store the buffer blocks in a region at the end
(called the ``buffer region''), storing the elements in an unordered way.
Indices into the buffer region
can shift as the region moves/grows/shrinks during updates,
thus we employ
a partial-sum data-structure of~\Cref{lemma:partial-sum-small}
using $\O((n/\log n)\cdot \log n)=\O(n)$ bits.
See~\Cref{fig:memory}.

\begin{figure}[ht]
\centering
\begin{tikzpicture}[scale=0.38]
\foreach \x in {0,1,...,8}
    \draw (\x-10,0) rectangle ++(1, 1);
\foreach \x in {0,1,...,8}
    \draw (\x+10,0) rectangle ++(1, 1);
\foreach \x in {0,1,...,8}
    \draw (\x, 0) rectangle ++(1, 1);
    
\foreach \x in {0,1,2}
    \draw (\x-4+3,-3) rectangle ++(1, 1);
\draw[ultra thick] (0-4+3,-3) rectangle ++(3, 1);
\foreach \x in {0,1,2}
    \draw (\x+3,-3) rectangle ++(1, 1);
\draw[ultra thick] (0+3,-3) rectangle ++(3, 1);
\foreach \x in {0,1,2}
    \draw (\x+4+3,-3) rectangle ++(1, 1);
   
\draw (-10,-6) rectangle ++(22, 1);
\draw (12,-6) rectangle ++(7, 1);

\foreach \x in {0,1,2}
    \draw (\x+5,-6) rectangle ++(1, 1);
\draw[ultra thick] (0+5,-6) rectangle ++(3, 1);
\draw[ultra thick,->] (0-4+3+3/2,-3) -- (0+5+3/2,-5+0.1);
\foreach \x in {0,1,2}
    \draw (\x-4,-6) rectangle ++(1, 1);
\draw[ultra thick] (0-4,-6) rectangle ++(3, 1);
\draw[ultra thick,->] (0+3+3/2,-3) -- (0-4+3/2,-5+0.1);

\draw[ultra thick] (13,-6) rectangle ++(1, 1);
\draw[ultra thick,->] (1+4+3+1/2,-3) -- (13+1/2,-5+0.1);
\draw[ultra thick] (16,-6) rectangle ++(1, 1);
\draw[ultra thick,->] (2+4+3+1/2,-3) -- (16+1/2,-5+0.1);
\draw[ultra thick] (18,-6) rectangle ++(1, 1);
\draw[ultra thick,->] (0+4+3+1/2,-3) -- (18+1/2,-5+0.1);

\foreach \x/\y in {0/2, 1/5, 2/8, 3/10, 4/1, 5/6, 6/0, 7/9, 8/4}
    \draw[ultra thick,->] (\x+1/2,0) -- (\y-1/2,-2+0.1);

\draw[ultra thick, <->] (0,1+1/2) -- ++(9, 0);
\draw (0+9/2,1+1/2) node[anchor=south] {Superblock};
\draw (-2,-3+1/2) node[anchor=east] {Blocks:};
\draw[dashed] (6+1/2,-3-1/2) rectangle ++(12, 2);
\draw (11,-3+1/2) node[anchor=west] {Buffer blocks};
\draw[ultra thick, <->] (-10,-6-1/2) -- ++(22, 0);
\draw (-10+22/2,-6-1/2) node[anchor=north] {Non-Buffer region};
\draw[ultra thick, <->] (12,-6-1/2) -- ++(7, 0);
\draw (12+7/2,-6-1/2) node[anchor=north] {Buffer region};

\draw (-10-1/2,0+1/2) node[anchor=east] {$A=$};
\draw (-10-1/2,-6+1/2) node[anchor=east] {$A^*=$};

\end{tikzpicture}
\caption{Memory organization of~\Cref{thm:sequence}.}
\label{fig:memory}
\end{figure}
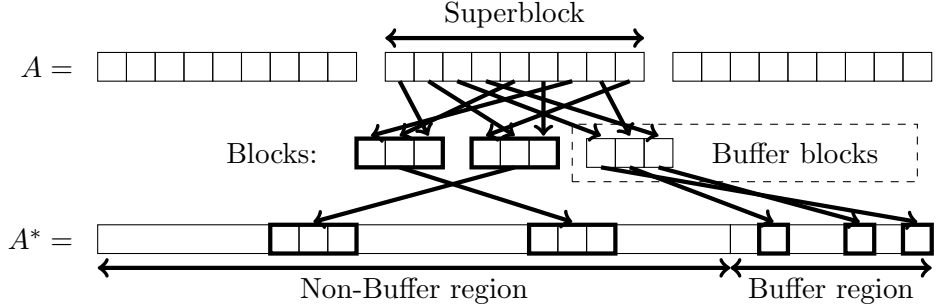

It takes $\O((n/\log^2 n)\cdot\log n)=\O(n/\log n)$ bits to store the permutation of the non-buffer blocks in $A^*$.
The buffer region at the end does not store its elements contiguously,
and instead uses~\Cref{lemma:previous},
taking $\O((n/\log^4 n)\cdot \log^2 n\cdot \log n)=\O(n/\log n)$ bits.
Within each superblock, we just need to store the permutation of the elements across the blocks.
Directly encoding the permutation as before uses $\O(\log(\log^4 n))=\O(\log\log n)$ bits per element.
We will improve this soon using recursion on $k$, but note we already achieved $\O(n\log\log n)$ redundancy.

For an insertion or deletion, we will handle it in the buffers.
To insert an element, we add it to the last (buffer) block of its superblock.
If we delete an element in a non-buffer block,
we move an element from the last (buffer) block.
Every time we create/destroy a buffer block after many insertions/deletions,
we move $\log^2 n$ elements to/from the buffer region as a block.
For an insertion, we do this by moving it to the end of non-buffer region
and shifting the indices of the buffer region; deletion is similar.
To de-amortize this to get worst-case $\O(1)$ time, we perform lazy updates.
Finally,
if the superblocks get too big/small (resp. double or half the size),
we merge/split appropriately.
This part only updates the indexing data-structure, not $A^*$.

Now, we will improve the space to $\O(n\log^{(k)}n)$ additional bits and $\O(k)$ time.
We will use recursion on $k$; recall that we know the construction for $k=1$.
Assume that there is a construction with $\O(n\log^{(k-1)}n)$ additional bits and $\O(k-1)$ time.
If we look at the task of organizing the elements within a superblock,
this is again a dynamic sequence.
Looking at an individual superblock,
we employ the recursive organization for parameter $k-1$.
The respective external array is laid out in memory as the non-buffer blocks and the buffer elements;
it grows/shrinks in the last buffer block and reads/writes take $\O(1)$ time.
The indexing structure takes $\O(n\log^{(k-1)}(\log^4 n))=\O(n\log^{(k)} n)$ bits to manage each superblock.
The additional structures of the blocks and superblocks take $\O(n\sum_{j=1}^k 1/\log^{(j)} n)=\O(n)$ space, adding each level.
Therefore, the indexing structure uses $\O(n\log^{(k)}n)$ bits and takes $\O(k)$ time to insert/delete/access.
All elements in the sequence are stored in $A^*$ using $n\lg u$ bits and the result follows.
\end{proof}

To apply this to a $\SELECT$ dictionary,
the construction of~\cite{domingues2026compressing} stored the elements directly in a dynamic sequence, as the gap from $n\lg u$ to $\lg\binom{u}{n}$ was already $\O(n\log n)$.
To remove this and achieve a tighter redundancy, we use the Elias-Fano encoding~\cite{elias1974efficient,fano1971number,carter1978exact,pibiri2017dynamic,pibiri2020succinct}.
This encodes the dictionary as an array $V$ of $n$ $\lg(u/n)$-bit elements and
a bit-vector $H$ of $h\leq 3n$ bits (called the header).
Namely, for $\delta=\lg u - \lg n$, let:
\begin{align*}
V[i]&=\SELECT(S, i)\wedge 1^\delta\in\{0,1\}^\delta~\text{and}\\ 
H&=1^{m_0}01^{m_1}0\cdots 1^{m_{n-1}}0\in \brc{0,1}^{h},\text{ where }
m_i=\size{\set{x \in S}{(x\gg \delta)=i}}~.
\end{align*}
When $n$ doubles/halves, we lazily update $\delta$, paying $\O(n)$ bits.
To get the $i$-th element,
we need $\RANK/\SELECT$ on the bit-vector.
That is, $\RANK(H,i)=\sum_{j\leq i}H[j]$
and $\SELECT(H,i)$ finds the $i$-th one in $H$.
We compute:
\begin{equation}
\SELECT(S,i)=V[i]\oplus \brk{\RANK(\neg H, \SELECT(H, i)) \ll \delta}~.
\end{equation}
Note $\RANK(\neg H, x)=x+1-\RANK(H,x)$.
Thus, we only need to implement $\RANK/\SELECT$ on the bit-vector
and updates that add or remove bits from $H$.
We prove the following, inspired by~\cite{clark1996efficient} to reduce $\RANK/\SELECT$ to small lookup tables.

\begin{lemma}\label{lemma:bit-vector}
There is a data-structure encoding a bit-vector $H\in\brc{0,1}^n$,
using $\O(n)$ bits,
with  $\O(\log_{w}n)$ time $\RANK/\SELECT$ and insertion/deletion of bits.
\end{lemma}
\begin{proof}
We divide $H$ into segments of $\Theta(w)$ bits.
Using a balanced tree as described in~\citep[{\S3}]{domingues2026compressing},
we reduce $\RANK/\SELECT$ to a segment in $\O(\log_w n)$ time, using $\O((n/w)\cdot w)=\O(n)$ bits.
Let $b=\frac{1}{2}\lg w$.
We further split the segment into blocks of $\Theta(b)$ bits;
$\RANK/\SELECT$ can be computed in $\O(1)$ time using lookup-tables. 
We use lookup-tables $f:\brc{0,1}^{b}\times [b]\to[b]$
for $\RANK/\SELECT$ over a $b$-bit bit-vector;
these lookup-tables take $\O(\sqrt{w}\log^2 w)$ bits.
We merge/split the blocks appropriately to maintain the size invariant.
Using the partial-sum data-structure of~\Cref{lemma:partial-sum-small}
(for block sizes and counting zeros/ones),
a bit-vector with $\O(w)$ bits needs only $\O(w)+\O(\lg w\cdot w/b)=\O(w)$ bits to navigate to the right block and compute $\RANK/\SELECT$ in $\O(1)$ time.
\end{proof}

Therefore, we use~\Cref{thm:sequence} and~\Cref{lemma:bit-vector} with Elias-Fano encoding to get the following.

\begin{theorem}\label{coro:select}
For every $1\leq k\leq \log^* n$,
there is a $\SELECT$ dictionary for $n=w^{\O(1)}$ elements
using $\O(n\log^{(k)}n)$ redundancy
and taking $\O(k)$ time for $\SELECT$ and update.
\end{theorem}

\subsection{Small universe}

Now, we investigate the case where $\log u=\O(\sqrt{w})$,
which will be useful in the implementation of $\RANK$ for the general case.
We can achieve a much more succinct dictionary in this regime,
which shows that we cannot easily remove
the assumption that $\log u=\Theta(w)$ in
the lower bound of~\Cref{thm:lower-bound}.

\begin{theorem}\label{thm:small-universe}
For $\log u=\O(\sqrt{w})$,
there is a $\RANK/\SELECT$ dictionary for $n=w^{\O(1)}$ elements
using $\O(n)$ redundancy
and taking $\O(1)$ time for $\RANK/\SELECT$ and update.
\end{theorem}
\begin{proof}
We consider the same framework as~\Cref{thm:sequence} that encodes a sequence $A$ as an out-of-order array $A^*$ plus auxiliary data-structures to encode the permutation.
Let $s=w^{1/4}$.
We split the sequence into 
superblocks of $\Theta(s^3)$ elements,
and each superblock
into blocks of $s^2$ consecutive elements,
with $\Theta(s)$ elements separating the blocks (except the last block).
We call the consecutive elements that are between two consecutive blocks the ``separators''.
Note there are $\Theta(s)$ blocks per superblock
and $\Theta(s^2)$ separator elements per superblock.
Because $\log u=\O(\sqrt{w})=\O(w/s^2)$, a block fits within a word,
we can, in constant time, compute $\RANK$ and insert/delete elements from a block.
Therefore, we keep the elements within each block sorted.

The memory organization is similar to before.
The separators are packed into $\Theta(1)$ blocks;
we also keep the separator elements sorted.
The regular blocks and the separator blocks are fixed-size ($s^2$ elements) except the last block (which we can think of as the buffer blocks of~\Cref{thm:sequence}).
We store the fixed-size blocks of all superblocks first.
Then, we store the elements in the last block at the end using an arbitrary permutation (like the buffer region).
The permutation for elements in the buffer region takes $\O((n/s) \cdot \log w)=\O(n)$ additional bits, and similarly for the permutation of the blocks.
See~\Cref{fig:block}.

\begin{figure}[ht]
\centering
\begin{tikzpicture}[scale=0.4]
\foreach \x in {0,1,...,21}
    \draw (\x+3, 0) rectangle ++(1, 1);
\foreach \x/\y in {0/2,3/6,10/1,12/6,19/3,23/4}
    \foreach \i in {1,...,\y}
        \draw (\x+\i,-3) rectangle ++(1, 1);
\draw[ultra thick] (4,-3) rectangle ++(6, 1);
\draw[ultra thick] (13,-3) rectangle ++(6, 1);
\draw[ultra thick] (24,-3) rectangle ++(4, 1);
\foreach \i in {1,...,2}
    \draw[ultra thick, ->] (\i+2+1/2, 0) -- ++(-2,-2+0.1);
\foreach \i in {1,...,6}
    \draw[ultra thick, ->] (\i+4+1/2, 0) -- ++(-1,-2+0.1);
\foreach \i in {1,...,6}
    \draw[ultra thick, ->] (\i+11+1/2, 0) -- ++(1,-2+0.1);
\foreach \i in {1,...,4}
    \draw[ultra thick, ->] (\i+20+1/2, 0) -- ++(3,-2+0.1);
\draw[ultra thick, ->] (11+1/2, 0) -- ++(0,-2+0.1);
\foreach \i in {1,...,3}
    \draw[ultra thick, ->] (\i+17+1/2, 0) -- ++(2,-2+0.1);
    
\foreach \x/\y in {1/6,2/6,3/6,4/4} {
    \draw[ultra thick] (7*\x+1-\y,-6) rectangle ++(\y, 1);
    \foreach \i in {1,...,\y}
        \draw (7*\x+1-\y+\i-1,-6) rectangle ++(1, 1);
}
\draw[ultra thick,->] (7,-3) -- ++(-2,-2+0.1);
\draw[ultra thick,->] (16,-3) -- ++(-4,-2+0.1);
\draw[ultra thick,->] (2,-3) -- ++(15,-2+0.1);
\draw[ultra thick,->] (12-1/2,-3) -- ++(7,-2+0.1);
\draw[ultra thick,->] (22-1/2,-3) -- ++(-1,-2+0.1);
\draw[ultra thick,->] (26,-3) -- ++(1,-2+0.1);

\draw[ultra thick,<->] (3,1+1/2) -- ++(22, 0) node[midway, anchor=south]{Superblock};
\draw[ultra thick,<->] (2,-6-1/2) -- ++(6, 0) node[midway, anchor=north]{Block};
\draw[ultra thick,<->] (9,-6-1/2) -- ++(6, 0) node[midway, anchor=north]{Block};
\draw[ultra thick,<->] (16,-6-1/2) -- ++(6, 0) node[midway, anchor=north]{Separator block};
\draw[ultra thick,<->] (25,-6-1/2) -- ++(4, 0) node[midway, anchor=north]{Trailing block};
\end{tikzpicture}
\caption{Depiction of blocks in~\Cref{thm:small-universe}.}
\label{fig:block}
\end{figure}
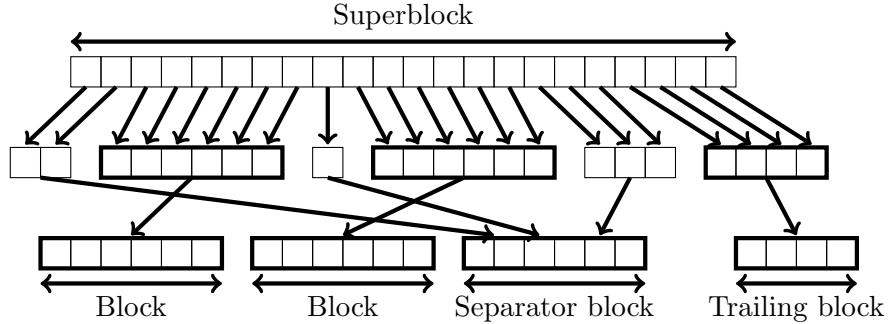

For $\SELECT$, we keep a partial-sum data-structure to find to which block/separator an index belongs, and this also
takes $\O((n/s) \cdot \log w)=\O(n)$ additional bits and $\O(1)$ time.
In fact, this data-structure can encode any dynamic sequence.
For $\RANK$,
we employ~\Cref{thm:domingues} with $\eps=1/4$
for the set of first and last elements of each block.
We compute $\RANK$ in this set in constant time using $\O((n/s^2)\cdot w^{1/4})=\O(n)$ additional bits.
With this, we determine one of two cases: either there is a predecessor element that is in a separator, or we need to compute $\RANK$ in the block.
Because the block is sorted and fits into a word, we can compute $\RANK$ in $\O(1)$ time, and similarly for the separator block.
For the trailing block, we just employ the same $\RANK$ index, using $\O((n/s)\cdot w^{1/4})=\O(n)$ additional bits.

When inserting into a block, we move the largest element to the separator, and shift the elements forward.
Similarly, when deleting an element in the block, take the consecutive element from the separators and shift the elements back.
At every insertion/deletion into the separator blocks, we move one element from the last separator to/from the trailing block.
The trailing block is allowed to grow and, when it reaches $2s^2+\Theta(s)$ elements, we create a new block.
The update for superblocks is analogous to~\Cref{thm:sequence}.

Now, we need a clean-up setup to maintain the size of the separators, while keeping the invariants of the data-structure.
Namely, we want that the number of elements in a separator is between $4s$ and $8s$.
We describe updates initially in an amortized fashion; recall that we have $\O(1)$ time access to all the separators.
When a separator between two blocks gets too large, 
we take the largest $s$ elements from the separators, and move them to the next block, displacing the largest $s$ elements from it, which go to the next separator,
all in constant time.
Because both the block and the separator are stored as packed arrays that fit in $\O(1)$ words,
this is just a sequence of bit-shifts and thus takes $\O(1)$ time per block.
We also need $\O(1)$ time to update the $\RANK$ index for first and last element of the blocks.
Similarly, when the size of the separator is too small,
we borrow $s$ elements from the next block, and so on.
In both cases, it takes $\Theta(s)$ time to go through all blocks.

By repeating this procedure in the superblock $\O(1)$ times,
we can make the number of elements in the separator always between $6s$ and $7s$ after the clean-up.
Importantly, the clean-up runs continuously, not triggered when some separator is too large/small.
When de-amortizing this,
we choose large enough constants so that
we complete a clean-up of the entire superblock in $s$ operations.
Then, all separators
(except for the one being currently resized)
have between $3s$ and $9s$ elements
while we are lazily performing this clean-up procedure,
and between $4s$ and $8s$ elements at the end of the clean-up.
Therefore, we keep the memory organization described previously with worst-case $\O(1)$ time for updates.

To conclude,
as in~\Cref{coro:select}, we use the Elias-Fano encoding, employing this construction for a dynamic sequence of $\delta=\lg u-\lg n$ bits.
The important observation is that the main data-structure (ignoring the $\RANK$ indices) can encode any dynamic sequence,
which is what we need for the Elias-Fano suffix.
With this, $\RANK/\SELECT$ in $S$ is reduced to $\RANK$ and indexing into this sequence.
Using the header in the Elias-Fano encoding, we can still solve $\RANK$ in a block in constant time, as we can restrict the block to only those that share the $(\lg n)$-bit prefix.
The result follows.
\end{proof}

We add the following mechanism to this construction,
which will be useful in extending this to the general universe case.
We divide a superblock (except for the trailing block) into ``miniblocks'' of between $s$ and $2s$ elements.
Applying~\Cref{lemma:partial-sum-small}, we can find the appropriate miniblock in $\O(1)$ time, using $\O((n/s)\log w)=\O(n)$ bits.
A miniblock can be stored in at most 2 blocks and 2 separator blocks.

\begin{observation}\label{obs:access}
In the construction of~\Cref{thm:small-universe},
every element is either in the trailing block,
or in a miniblock of $\Theta(s)$ consecutive elements.
A miniblock can be read in $\O(1)$ time.
\end{observation}

\begin{remark}
The previous construction of~\Cref{thm:small-universe} works with $\log u=\O(w/\log^2 w)$ in the cell-probe model using the $\RANK$ index of~\cite{ajtai1984hash}.
The argument follows {mutatis mutandis} replacing $s=\log w$.
\end{remark}

\subsection{Supporting rank}

To support $\RANK$, we slightly modify the $\SELECT$ dictionary, namely we exploit redundant bits in the sequence of elements to store the $\RANK$ index.
This differs from the techniques of~\cite{domingues2026compressing},
where there is a clear separation from the $\SELECT$ dictionary and the $\RANK$ index.

We use the $\SELECT$ dictionary from~\Cref{coro:select}.
Recall that we store the first $\lg n$ bits in an Elias-Fano header,
and the remaining $\delta=\lg u-\lg n$ bits as a sequence.
Let $\ell$ be a parameter we will set later.
To help with $\RANK$, we store the first $\lg n+\ell$ bits separately in the data-structure of~\Cref{thm:small-universe},
and the remaining bits as the sequence in~\Cref{thm:sequence}.
We note the following central observation.

\begin{observation}\label{lemma:exploit}
For $1\leq \ell< \delta$ and
a set $S$ of elements of $\delta$ bits,
one of the following holds:
\begin{enumerate}[label=(\Roman*)]
\item $\RANK$ in $S$ can be determined by only the prefix of $\ell$ bits plus one $\SELECT$ access of $S$, or
\item there are two elements of $S$ that share the prefix of $\ell$ bits.
\end{enumerate}
An analogous statement holds for multisets and $\RANKM$.
\end{observation}
\begin{proof}
If (II) holds, we are done.
Otherwise, all $\ell$-bit prefixes of $S$ are distinct, and thus for $r=\RANK(\set{x \gg (\delta-\ell)}{x\in S}, y \gg (\delta-\ell))$, we get $\RANK(S,y)\in\{r-1,r\}$, which can be determined by comparing $y$ to $\SELECT(S,r)$.
\end{proof}

We use this and~\Cref{obs:access} to prove~\Cref{thm:dictionary}.

\ThmDict*
\begin{proof}
Let $\ell=\Theta(\sqrt{w})$ and $s={w}^{1/4}$.
First, we use~\Cref{thm:small-universe} for a small universe of $\ell+\lg n$ bits to store the prefix (most-significant bits) of $\ell+\lg n$ bits from the elements of $S$.
Note that this also subsumes the Elias-Fano encoding.
Importantly, the main data-structure of~\Cref{thm:small-universe} supports repeated values (i.e. $\RANK/\SELECT$ on multisets).
Recall that the construction of~\Cref{thm:small-universe}
divides the prefixes into miniblocks of $\Theta(s)$ elements of $\ell$ bits.
The remaining bits of the elements are stored using~\Cref{thm:sequence}.
It takes additional $\O(n)$ bits and $\O(k)$ time to navigate to the right miniblock (note that the $\RANK$ indices from before now have to access the larger $\SELECT$ dictionary).

Now, we show how to compute $\RANK$ within each miniblock.
Recall that we can read an entire miniblock of~\Cref{thm:small-universe} in $\O(1)$ time (cf.~\Cref{obs:access}).
For each miniblock,
we use an indicator bit-vector to identify the elements with a repeated prefix; stored separately in the indexing structure of~\Cref{thm:small-universe}.
We mark with $1$ the prefixes that are seen again, reading first-to-last in the miniblock.
This adds $\O(n)$ bits overall.
We consider the two cases of~\Cref{lemma:exploit}, both of which can be computed in $\O(1)$ time.
\begin{enumerate}[label=(\Roman*)]
\item In this case, $\RANK$ can be decided by first looking at the first $\ell$ bits and, reading the entire array of $\Theta(s)$ elements, we compute $\RANK$ on the prefixes in $\O(1)$ time, as they fit into $\O(1)$ words  (there are no repeated values).
Then, performing one $\SELECT$ on the set is enough to compute the $\RANK$ in $S$ and takes $\O(k)$ time.
\item
For a sequence of elements with a repeated prefix,
we replace the bits in the second element by
the $\RANK$ index for this sequence using~\Cref{thm:domingues} for $\eps=1/4$,
which uses $\O(s\cdot w^{1/4})\leq\ell$ bits (we choose the constant in $\ell$ so that this inequality holds).
Recall that the miniblocks are sorted arrays.
To compute $\RANK$ in $S$,
first compute $\RANK$ in the prefix, then check if the
prefix is repeated.
If so, read the $\RANK$ index for elements sharing this prefix and compute the $\RANK$ there, accessing the $\SELECT$ dictionary.
We emphasize that we mask out the repeated prefixes when computing $\RANK$ in the miniblock.
Otherwise, we conclude the $\RANK$ as in (I).
\end{enumerate}
Given the encoding of Representation (II),
we can lazily split/merge miniblocks in constant time.

The only remaining issue is to concatenate the two out-of-order arrays of~\Cref{thm:sequence} and~\Cref{thm:small-universe}.
We solve this by a kind of interleaving.
That is, we concatenate
every contiguous sequence of
$\brk*{\sqrt{w}/\log^2 n}$ blocks of~\Cref{thm:sequence} ($\log^2 n$ elements of $\delta-\ell$ bits)
with a block of~\Cref{thm:small-universe} ($\sqrt{w}$ elements of $\ell=\O(\sqrt{w})$ bits).
Because the blocks of~\Cref{thm:small-universe} fit into a word, we can shift them in the last combined block in $\O(1)$ time.
In total, this construction adds only $\O(n)$ bits to the space and only $\O(1)$ more accesses, compared to~\Cref{coro:select}.
\end{proof}

\subsection{Extension to FIDs}

In~\cite{domingues2026compressing},
supporting $\RANK/\SELECT$ in the complement of $S$ is also considered.
A dictionary that supports $\RANK/\SELECT$ for both $S$ and $\overline{S}$ is called a fully indexable dictionary (FID).
Note that $\RANK(\overline{S}, x)= x+1-\RANK(S,x)$, thus we focus on $\SELECT(\overline{S}, i)$.
We extend~\Cref{thm:dictionary} to this case as well,
using the techniques from~\cite{grossi2009more} and~\citep[{\S4}]{domingues2026compressing}.
The idea is to use a partial-sum data-structure.

\begin{observation}[\cite{grossi2009more}]
Let $Z[i]=\SELECT(S,i+1)-i$, for $0\leq i\leq n-1$.
Then for every $i$, the following holds: $\SELECT(\overline{S},i+1)=i+1+\RANKM(Z, i)$.
\end{observation}
\begin{proof}
The idea is that $Z[i]$ is the number of zeros before the $(i+1)$-th one.
Moreover, $Z$ is non-decreasing.
By definition, $\RANKM(Z, i)=\max\set{0\leq j\leq n-1}{i\geq \SELECT(S,j+1)-j}$.
For $0\leq i\leq n-1$, assuming $\RANKM(Z, i)\neq -1$, we write:
\begin{align*}
z=i+1+\RANKM(Z, i)&=\max\set{i+j+1}{i+j\geq \SELECT(S,j+1)}\\
&=\max\set{z'\geq i+1}{z'-1\geq \SELECT(S,z'-i)}~.
\end{align*}
Then $1+\SELECT(S,z-i)\leq z<\SELECT(S,z+1-i)$,
taking the convention that $\SELECT(S,n+1)=u$.
If $\RANKM(Z, i)=-1$, then $z=i<\SELECT(S,1)$.

Consider the characteristic bit-vector $B$ for the set $S$, that is, $B[i]=1$ iff $i\in S$.
The inequalities on $z$ imply that
the $(z-i)$-th one of $B$ is strictly below the index $z$ and
the $(z+1-i)$-th one is strictly above the index $z$.
Therefore, there are exactly $i+1$ zeros in $B[0:z]$ and $B[z]=0$,
thus $z=\SELECT(\overline{S},i+1)$.
\end{proof}

\ThmFID*
\begin{proof}
Recall~\Cref{lemma:partial-sum}, and observe that the accesses to the sum $Z$ are kept up-to-date by the $\SELECT$ dictionary.
If $\log u=\O(\sqrt{w})$,
then, it suffices to modify~\Cref{thm:small-universe} 
to compute $\RANKM$ on $Z$, where $Z[i]=\SELECT(S,i+1)-i$.
The crucial observation is that we read $\O(\sqrt{w})$ elements packed into a word
and compute $\RANKM$ in $Z$ in $\O(1)$ time.
Concretely, we compute $Z$ on the spot, using parallel subtraction on the packed array, after reducing to a common $\log n$ prefix using the Elias-Fano header (see~\cite{patrascu2014dynamic,domingues2026compressing}).
Then $\RANKM(Z,i)$ is constant-time on $Z$ as a packed array.
The $\RANK$ index for finding the block is replaced by
partial-sum data-structure of~\Cref{lemma:partial-sum} instantiated with~\Cref{thm:domingues} for $\eps=1/4$,
which uses $\O((n/s)\cdot w^{1/4})=\O(n)$ bits.
Overall, this achieves $\O(n)$ redundancy and $\O(1)$ time for all operations.

Otherwise, for $\lg u\geq 2\lg n+\ell=\O(\sqrt{w})$,
we proceed by
adapting the ideas of~\Cref{thm:dictionary},
(using Elias-Fano and the small universe) 
where we look at the prefix $\SELECT(S,i+1)\gg (\delta-\ell)$ where $\delta=\lg u-\lg n$ and $\ell=\Theta(\sqrt{w})$.
The main issue, compared to~\Cref{thm:dictionary}, is that the prefix of $\SELECT(S,i+1)$ may differ from $Z[i]=\SELECT(S,i+1)-i$.
The trick is to notice that since $i$ has few bits
(namely, $\log i\leq \log n \leq {\delta-\ell}$),
then $Z[i]\gg (\delta-\ell)$ will be either $p$ or $p-1$, where $p= \SELECT(S,i+1)\gg (\delta-\ell)$.
Therefore, we only need to check both prefixes and we can reuse all the $\RANKM$ machinery used in~\Cref{thm:dictionary} for the prefixes of~\Cref{thm:small-universe}.

Recall that we have a $\RANK$ index for the separators and we can navigate to the relevant miniblock in $\O(k)$ time.
Again, we consider the two cases of~\Cref{lemma:exploit} for $\ell=\Theta(\sqrt{w})$ and a multiset (a miniblock) of $\Theta(w^{1/4})$ elements.
\begin{enumerate}[label=(\Roman*)]
\item Because the miniblock fits into one word, we can compute $\RANKM$ on the prefixes in constant time, plus $\O(1)$ accesses to the $\SELECT$ dictionary. If all prefixes are distinct, we get a unique answer, and we only have to check this element and the previous (due to the difference to the prefix $Z$).
\item We replace the bits in the shared prefix in one of the elements by
the partial-sum structure, taking $\O(s\cdot w^{1/4})\leq \ell$ bits.
This is analogous to~\Cref{thm:dictionary}.
\end{enumerate}
Therefore, we support $\RANK/\SELECT$ on the complement of $S$
(for $\size{S}=n=w^{\O(1)}$ elements)
using asymptotically the same redundancy and time as~\Cref{thm:dictionary}.
\end{proof}

\section{Lower bound}

In this section, we prove our main lower bound,
based on the proof of~\cite{li2023tight},
modified to handle $u\geq n^2$.
The techniques of~\cite{li2023tight} consider
a universe $u=n^{1+\alpha}$ where $\alpha=\Theta(1)$.
Removing this requirement is described in the mechanized proof of~\cite{dambrosio2026always} focusing on the case $R<n$.
We provide here a simplification of the techniques that suffices for our regime of $R\geq n$,
while side-stepping a technical issue pointed out by~\cite{dambrosio2026always}.

Recall that the proof considers the following hard instance:
\begin{itemize}
\item Let $K$ and $K'$ be two uniformly random disjoint subsets of $[u]$ of $n$ elements.
\item  Let $D$ and $E$ be uniformly random permutations of $K$ and $K'$, respectively.
\end{itemize}
The hard instance inserts all elements of $K$ into the dictionary, then, for every $i$ in $[n]$, the instance performs a \emph{meta-operation}: queries $D[i]$, deletes $D[i]$, and inserts $E[i]$.

Next, we consider a sequence of $m$ meta-operations
that are divided into segments.
A probe to a cell $c$ is a \emph{repeated probe} if:
\begin{enumerate*}[label=(\roman*)]
\item the cell $c$ is probed in multiple segments, and
\item it is the first probe to $c$ in the segment.
\end{enumerate*}
The proof relies on the following central lemma.

\begin{restatable}[Outer Lemma]{lemma}{LemmaOuter}\label{lemma:outer}
For $\log u=w$ and $n\leq \sqrt{u}$,
consider a sequence of $m\leq n$ meta-operations, divided into $\lambda$ segments.
If $m\geq \Cfactor \max\brc{R,n}/\log \lambda$ and $\Clambda\leq\lambda\leq \frac{\log n}{\log\log n}$,
then, for $u$ sufficiently large, the expected number of repeated probes is at least $m/\Cprobe$.
\end{restatable}

Our main departure compared to~\cite{li2023tight,dambrosio2026always} is to consider only the number of repeated probes,
simplifying the proof and removing the polynomial universe requirement.
Moreover, we slightly optimize the constants.
See the proof of this lemma in~\Cref{appendix:outer}.
Given this, we prove the main lower bound, and thus~\Cref{thm:lower-bound} follows.
\LemmaProbe*
\begin{proof}
Divide the $n$ meta-operations into segments of $128\max\{R,n\}/\log n$ operations; construct a tree, where  each of those segments is a leaf.
Level $i\geq 1$ has degree $\lambda_i=\log^{(i)}n/\log^{(i+1)} n$.
We stop the tree and put a root if either 
we connected all meta-operations or
$\lambda_{i}< \Clambda$.
The node at level $i$ is responsible for a segment of $m_i=128\max\{R,n\}/\log^{(i+1)} n$ operations.
Because $R=\O(n\log^{(k)}n)$,
the tree has height $h=\Theta(k)$.

For every node at every level $i$ (except the root),
we apply the Outer Lemma (\Cref{lemma:outer}).
The requirements are
that $1024\leq \lambda_i\leq \lambda_1={\log n}/{\log\log n}$
and that
$m_i\geq \Cfactor\max\{R,n\}/\log \lambda_i$,
which hold for our case.
Thus, the number of repeated probes in the level $i$ is at least $n/\Cprobe$.
Except for the first repeated probe to cell $c$ in the node,
each repeated probe to cell $c$
is counted uniquely by the Outer Lemma (not counted in the upper levels).
Removing the first repeated probe, the remaining probes are at least half of the repeated probes counted in a given node.
To account for this,
we underestimate by charging $\frac{1}{2}$ for each repeated probe
when summing over all levels.
This yields at least $\frac{1}{2}(nh/\Cprobe)=\Omega(nk)$ probes for the sequence of meta-operations.

To extend to $\log u=\Theta(w)$,
we think of the memory as cells of $(\lg u)$-bits.
Applying the theorem for $(\lg u)$-bit probes,
then the number of $w$-bit probes is at least $\frac{1}{2}(nh/\Cprobe)/(\ceil{w/\lg u}+1)=\Omega(nk)$.
\end{proof}

\section{Discussion and Open problems}

In this paper,
we have constructed a space-efficient data-structure for dynamic $\RANK/\SELECT$ dictionaries in the standard Word-RAM model and proved that,
in the regime $n\leq w^{\O(1)}$,
the running time for the operations is optimal for the redundancy used.
Moreover, we provide a generalization of the lower bound of~\cite{li2023tight} for dynamic dictionaries in super-polynomial universes with a simpler proof that shows our upper bound is optimal.
There are two remaining open problems, which we will discuss next.

\paragraph*{Smaller universes.}
We have shown that we can achieve better results when $\log u=\O(\sqrt{w})$ in the standard Word-RAM model
and when $\log u=w/\polylog w$ in the cell-probe model.
This is a sharp distinction compared to the $\log u=\Theta(w)$ case of the lower bound.
We ask what exactly are the trade-offs of universe size to redundancy when $\log u=o(w)$.

\paragraph*{Achieving $o(n)$ redundancy.}
To our knowledge, no deterministic data-structure for dynamic ordered dictionaries achieved $o(n)$ redundancy for any regime.
The recent result of~\cite{kuszmaul2026succinct} used a randomized data-structure to achieve $\O(nw/2^{(\log_w n)^{1/3}})$ redundancy and optimal amortized time in expectation.
For the regime $n=w^{\omega(\log^4 w)}$, the redundancy is $o(n)$.
This can be de-amortized with the cardinality reduction lemma of~\cite{domingues2026compressing},
however, it incurs an additional $\O(n)$ redundancy.
Further work is required to achieve worst-case times for $o(n)$ redundancy.
The lower bounds show that we need at least $\Omega(\log^* n + \log(n/R))$ time for $R<n$ redundancy.
This has been achieved for dictionaries that only support membership~\cite{li2024dynamic}.


\section*{Acknowledgments}
This research was supported by the Israel Science Foundation, grant No. 1948/21.
The author thanks the anonymous reviewers for their helpful comments.

\paragraph*{AI disclosure.} All technical content and text in this paper were written by the author. Claude Opus 5.5 was used to proofread the text, improve the clarity of exposition, and point out typos and inconsistencies.

\bibliography{main}

\begin{appendices}

\section{Outer Lemma}\label{appendix:outer}

We work towards a proof for the Outer Lemma (\Cref{lemma:outer})
following closely the proof of~\cite{li2023tight}, while removing the requirement that $u=n^{\O(1)}$.
The structure of the proof is as follows: assuming that the expected number of repeated probes is at most $m/\Cprobe$, show a protocol for a communication game that contradicts the lower bound on $m$.
Since we are interested in the range $R\geq n$,
we are able to simplify the proof to only consider the number of probes in different segments.

The communication game is performed as follows.
Let $E$ be the insertion sequence, that is divided into $\lambda$ segments of $\mu=m/\lambda$ operations denoted $E_i$.
Let $A_i$ be the set of keys of $E_i$.
Similarly, the deletion sequence $D$ is divided into segments $D_i$.
The proof considers a communication game between Alice and Bob,
where Bob knows:
\begin{enumerate}[label=(\roman*)]
    \item the initial memory state $C_{\text{init}}$,
    \item the set of insertion keys $A=\bigcup_{i=1}^{\lambda}A_i$,
    \item the deletion sequence $D$,
    \item a set of permutations $\pi_i$ such that $\pi_i(A_i)=E_i$ (interpreting $A_i$ as a sorted sequence).
\end{enumerate}
Given this, Bob wants to learn all $A_i$.
The entropy, given what Bob knows, 
is $\log\brk*{m!/(\mu!)^\lambda}$,
that is, the entropy of arranging a set of $m$ elements into $\lambda$ segments of $\mu$ elements (without order in each segment).

First, Alice sends to Bob $C_{\text{end}}$, the final memory state.
Since Bob knows $C_{\text{init}}$ and the set of insertion keys and the deletion sequence, he knows the final set of keys $K_{\text{end}}$.
Thus, Alice only needs to send $H(C_{\text{end}}\mid C_{\text{init}}, D,A,\pi)\leq H(C_{\text{end}})-H(K_{\text{end}})\leq\log\binom{u}{n}+R-\log\binom{u}{n}=R$ bits.

A memory state is a partial map $C:\{0,1\}^w\rightharpoonup \{0,1\}^w$ and $C(S)$ is the memory state of the cells in the set $S$.
The distance between two memory states is the number of cells that they differ.
Starting at a state $C$, let $C_{D,B}$ denote the state after performing the sequence of meta-operations: for $j\in[\mu]$, query $D[j]$, delete $D_i[j]$, and insert $\pi_i(B)[j]$.
Denote by $S_{D,B}$ the set of cells probed during this sequence of meta-operations.

For every $1\leq i\leq \lambda$, we proceed with the following protocol.
Let us refer to the cells that are probed in multiple segments as \emph{repeated}.
Let $X_i=1$ if segment $i$ probes fewer than $\mu/\Cgamma$ repeated cells, and $X_i=0$ otherwise.
Let $C_i$ denote the memory state before applying the $i$-th segment operations.
Define the following families of sets.
\begin{align*}
\mathcal{F}_i&=\set{B\subseteq \textstyle\bigcup_{j=i}^{\lambda}A_j}{\size{B}=\mu}\\
\mathcal{Q}_i&=\set{B\in \mathcal{F}_i}{\dist(C_{D_i,B}(S_{D_i,B}), C_{\text{end}}(S_{D_i,B}))<\mu/\Cgamma,~\text{where }C=C_i}
\end{align*}
Observe that if $X_i=1$, then $A_i\in\mathcal{Q}_i$.
Then,  Alice sends the message $M_i$: 
first send $X_i$,
then, if $X_i=1$, she sends the index of $A_i$ in $\mathcal{Q}_i$,
otherwise she sends the index of $A_i$ in $\mathcal{F}_i$.
Importantly, $B$ can simulate the meta-operations and thus knows $\mathcal{Q}_i$.

Now we bound the message length.
Denote $\mathcal{G}_i=\set{B\in \mathcal{F}_i}{\size{A_i\cap B}\geq \mu/2}$.
Let $Y_i=1$ iff $\mathcal{Q}_i\subseteq \mathcal{G}_i$.
Then, the length of the message $M_i$ is bounded by:
$M_i
=1+(1-X_i)\cdot\log\size{\mathcal{F}_i}+X_i\cdot\log\size{\mathcal{Q}_i}
\leq 1+(1-X_iY_i)\cdot\log\size{\mathcal{F}_i}+X_iY_i\cdot\log\size{\mathcal{G}_i}
= 1+\log\size{\mathcal{F}_i}-X_iY_i\cdot
\brk*{\log\size{\mathcal{F}_i}- \log\size{\mathcal{G}_i}}$.
To complete the proof of the Outer Lemma,
we bound the probability that two sets with bounded intersection
will have few cells that differ after the sequence of meta-operations.
Note that if $A,B\in\mathcal{Q}_i$, for $S=S_{D,A}\cap S_{D,B}$,
then
$\dist(C_{D,A}(S), C_{\text{end}}(S))<\mu/\Cgamma$
and
$\dist(C_{D,B}(S), C_{\text{end}}(S))<\mu/\Cgamma$,
thus 
$\dist(C_{D,A}(S), C_{D,B}(S))<\mu/\Cdelta$.
We show the following Inner Lemma (the proof is deferred to~\Cref{appendix:inner}).
\begin{restatable}[Inner Lemma]{lemma}{LemmaInner}\label{lemma:inner}
For $\log u=w$,
$n\leq \sqrt{u}$,
$\Cfactor\max\{R,n\}/\log n\leq \mu<n/2$,
and $g<\mu/2$,
consider the following distribution:
$K$ is a uniform subset of $[u]$ with $n$ elements,
$A$ and $B$ are uniform subsets of $[u]\setminus K$,
each with $\mu$ elements and having $\size{A\cap B}<\mu/2$,
and $D$ is a uniform sequence of $K$ with $\mu$ elements.
Then, for initial state $C$ with key-set $K$ and for sufficiently large $u$,
\begin{equation}
\Prc!{\dist(C_{D,A}(S_{D,A}\cap S_{D,B}), C_{D,B}(S_{D,A}\cap S_{D,B})) \leq \mu/\Cdelta}{\size{A\cap B}=g}\leq u^{-\mu/16}~.
\end{equation}
\end{restatable}

We required $m\geq \Cfactor \max\{R,n\}/\log \lambda$
and $\lambda\leq\frac{\log n}{\log\log n}$, and
thus $\mu\geq \Cfactor\max\{R,n\}/(\lambda\log \lambda)\geq \Cfactor\max\{R,n\}/\log n$.
Now we have all the ingredients to prove the Outer Lemma;
following the procedure of~\cite{li2023tight} with slightly different constants.

\LemmaOuter*
\begin{proof}
Suppose by contradiction that the expected number of repeated probes is at most $m/\Cprobe$.
Using this, we prove $\mathbb{E}\brk[s]{\sum_{i=1}^{\lambda/2}X_i}\geq \lambda/4$.
If $X_i=0$, then segment $i$ probes at least $\mu/\Cgamma$ repeated cells.
Thus:
$\mathbb{E}\brk[s]*{\sum_{i=1}^\lambda (1-X_i)\cdot {\mu/\Cgamma}}
\leq {m/\Cprobe}
\Rightarrow
\mathbb{E}\brk[s]*{\sum_{i=1}^\lambda (1-X_i)}
\leq {\lambda/4}$.
Hence, $\mathbb{E}\brk[s]*{\sum_{i=1}^{\lambda/2} X_i}\geq {\lambda/2}-\mathbb{E}\brk[s]*{\sum_{i=1}^{\lambda/2}(1-X_i)}\geq \lambda/2-{\lambda/4}=\lambda/4$.

For the next step,
we observe
$\Pr{Y_i=0\wedge X_i=1}
\leq \Pr{A_i\in\mathcal{Q}_i\text{ and }{\brk{\mathcal{F}_i\setminus \mathcal{G}_i}\cap \mathcal{Q}_i} \neq\varnothing}$,
recalling that if $X_i=1$, then $A_i\in\mathcal{Q}_i$.
Recall that if $A,B\in\mathcal{Q}_i$,
then
$\dist(C_{D,A}(S_{D,A}\cap S_{D,B}), C_{D,B}(S_{D,A}\cap S_{D,B}))<\mu/\Cdelta$.
We get the following bound:
\begin{align*}
\operatorname{Pr}[A_i\in\mathcal{Q}_i&\text{ and }
{(\mathcal{F}_{i}\setminus\mathcal{G}_i)\cap \mathcal{Q}_i}\neq\varnothing]\leq\\
&\leq \operatorname{Pr}[A_i\in\mathcal{Q}_i\text{ and }{\exists B\in (\mathcal{F}_{i}\setminus\mathcal{G}_i)}:B\in \mathcal{Q}_i]\\
&\leq \mathbb{E}[\size{\set{{B\in (\mathcal{F}_{i}\setminus\mathcal{G}_i)}}{A_i\in\mathcal{Q}_i\text{ and }B\in \mathcal{Q}_i}}]\\
&\leq \size{\mathcal{F}_{i}\setminus\mathcal{G}_i}\cdot
\max_{0\leq g<\mu/2}\Prc!{A,B\in \mathcal{Q}_i}{\size{A\cap B}=g}
\qquad\text{($A,B\in\mathcal{F}_i$)}
\\
&\leq \size{\mathcal{F}_i}\cdot\max_{0\leq g<\mu/2}
\Prc!{\dist(C_{D,A}(S_{D,A}\cap S_{D,B}), C_{D,B}(S_{D,A}\cap S_{D,B}))\leq \mu/\Cdelta}{\size{A\cap B}=g}
\end{align*}
We apply the Inner Lemma (\Cref{lemma:inner}) using the key-set $K_i=\brk*{K\setminus \bigcup_{j=1}^{i-1} D_j}\cup \brk*{\bigcup_{j=1}^{i-1} A_j}$ and the state $C_i$.
The subtle issue here is that we use the symmetry of distribution and the initial state that is present in the Inner Lemma.
Importantly, $K_i$ is a uniformly distributed $n$-element subset of $[u]$.
This yields $$\Pr{Y_i=0\wedge X_i=1}\leq \binom{m}{\mu}\cdot u^{-\mu/16}\leq\brk*{\frac{4\lambda}{u^{1/16}}}^{\mu}\leq \frac{1}{8}~.$$
for $u$ sufficiently large.
Hence, 
$\sum_{i=1}^{\lambda/2}\mathbb{E}\brk[s]{X_iY_i}=\mathbb{E}\brk[s]{\sum_{i=1}^{\lambda/2}X_i}-\sum_{i=1}^{\lambda/2}\Pr{X_i=1\wedge Y_i=0}\geq \lambda/4-(1/8)\cdot\lambda/2=3\lambda/16$.

Next, we bound $\log\size{\mathcal{F}_i}-\log\size{\mathcal{G}_i}$, for $i\leq\lambda/2$.
We note that in this case, 
$\size{\bigcup_{j=i}^\lambda A_j}\geq m/2$
and $\size{\bigcup_{j={i+1}}^\lambda A_j}\leq m-\mu$.
Therefore:
\begin{align*}
\log\size{\mathcal{F}_{i}}-
\log\size{\mathcal{G}_{i}}
&=\log{\binom{\size{\bigcup_{j=i}^\lambda A_j}}{\mu}}
-\log\brk[s]*{\sum_{g=\mu/2}^\mu \binom{\mu}{g}\binom{\size{\bigcup_{j=i+1}^\lambda A_j}}{\mu-g}}\\
&\geq\log{\binom{m/2}{\mu}} 
-\log\brk[s]*{\mu\binom{\mu}{\mu/2}\binom{m-\mu}{\mu/2}}\\
&=\brk[s]*{\log{\binom{m/2}{\mu}}-\log\binom{m/2}{\mu/2}}
-\log\brk[s]*{\mu\binom{\mu}{\mu/2}}
+\brk[s]*{\log{\binom{m/2}{\mu/2}}-\log\binom{m-\mu}{\mu/2}}
\\
&\geq\brk[s]*{\log{\binom{(m-\mu)/2}{\mu/2}}-\log\binom{\mu}{\mu/2}}
-\log\brk[s]*{\mu\binom{\mu}{\mu/2}}
-\frac{\mu}{2}
\\
&=\log{\binom{(m-\mu)/2}{\mu/2}}
-2\log{\binom{\mu}{\mu/2}}
-\log\mu
-\frac{\mu}{2}
\\
&\geq \frac{\mu}{2}\log(\lambda-1)
-2\mu
-\frac{\mu}{2}
\geq \frac{\mu}{2}\log(\lambda/64)~,
\end{align*}
where we used the inequalities
$\log\binom{u}{n}\geq n\log\frac{u}{n}$,
$\log\binom{u}{u/2}\leq u-\frac{1}{2}\log u$,
and
$\log\binom{u}{n}-\log\binom{u'}{n}\leq n$  
for $u\geq u'\geq (u+n)/2$.
Hence, we get the bound $\log\size{\mathcal{F}_i} - \log\size{\mathcal{G}_i}\geq ({\mu}/{2})\cdot \log(\lambda/64)$ for $i\leq\lambda/2$.

Lastly, we bound the number of bits Alice sends to the entropy of what Bob learns: $H(M_{\text{learn}}\mid C_{\text{init}}, A, D, \pi)\leq H(M_{\text{send}}\mid C_{\text{init}}, A, D, \pi)$.
Therefore, noting that $\prod_{i=1}^\lambda\size{\mathcal{F}_i}=m!/(\mu!)^\lambda$, we get:
\begin{align*}
\log(m!/(\mu!)^\lambda)
&\leq R +\sum_{i=1}^\lambda\mathbb{E}[M_i] 
=
R + \sum_{i=1}^\lambda\brk[s]!{1+\log\size{\mathcal{F}_i}-\mathbb{E}\brk[s]{X_iY_i}\cdot
\brk*{\log\size{\mathcal{F}_i}- \log\size{\mathcal{G}_i}}}
\\&\leq
R + \lambda + \log(m!/(\mu!)^\lambda) - \sum_{i=1}^{\lambda/2} \mathbb{E}\brk[s]{X_iY_i}\cdot
\brk*{\log\size{\mathcal{F}_i}- \log\size{\mathcal{G}_i}}
\\&\leq
R + \lambda + \log(m!/(\mu!)^\lambda) -\frac{3\lambda}{16}\cdot\frac{\mu}{2}\log(\lambda/64)~.
\end{align*}
Thus, $R\geq \frac{3}{32}m\log(\lambda/64)-\lambda\geq \frac{1}{32}m\log \lambda$ 
(using $\lambda\geq\Clambda$ and $\mu\geq 64$),
which is a contradiction to the requirement $m\geq \Cfactor R/\log\lambda$.
Therefore, it is impossible for the expected number of repeated probes to be at most $m/\Cprobe$.
\end{proof}

\section{Inner lemma}\label{appendix:inner}

Next, we prove the inner lemma.
The proof follows~\cite{li2023tight}, with the simplification for large $\mu$, that avoids using more complicated randomized set separators.
Due to this, we are able to drop some requirements, namely bounds on the expected size of $S_{D,A}$ and other inequalities in $n$ and $\mu$; this side-steps the conditioning issue pointed out in~\cite{dambrosio2026always},
that required an expectation bound on the number of probes.

\LemmaInner*
\begin{proof}
Let $W=1$ if $\dist(C_{D,A}(S_{D,A}\cap S_{D,B}), C_{D,B}(S_{D,A}\cap S_{D,B})) \leq \mu/\Cdelta$.
Again, the proof follows the structure of~\cite{li2023tight},
playing a communication game between Alice and Bob.
In this case, Bob only knows the permutations of the insertions and deletions, denoted $\pi$ and $\sigma$.
Bob wants to learn $A$, $B$, $D$, and $K$, if $W=1$.
The entropy of what Bob wants to learn is:
\begin{align*}
H(M_{\text{learn}}|W=1)&\geq
-\log\max_{A,B,D,K}\Prc{A,B,D,K}{W=1}\\
&\geq \log\Pr{W=1}
+\log\binom{u}{n}
+\log\binom{n}{\mu}
+\log\binom{\mu}{g}
+\log\binom{u-n-\mu}{\mu-g}
+\log\binom{u-n}{\mu}\\
&\geq \log\Pr{W=1}
+\log\binom{u}{n}
+\log\binom{n}{\mu}
+\log\binom{u-n-\mu}{\mu/2}
+\log\binom{u-n}{\mu}~.
\end{align*}

Alice sends $W$, $B$, a mixed memory state $C_{\text{mix}}$
where $C_{\text{mix}}(S_{D,B})=C(S_{D,B})$ and $C_{\text{mix}}(\overline{S_{D,B}})=C_{D,A}(\overline{S_{D,B}})$.
We observe $C_{\text{mix}}$ has at most $2n$ cells, as 
$$\brk*{\lg\binom{u}{n}+R}/\lg u\leq n+R/\lg u\leq n+\frac{n\log n}{128\lg u}\leq n+\frac{n}{256}~,$$
using the conditions
$n\leq \sqrt{u}$ and
$\Cfactor R/\log n<n/2$.
Then, Alice sends $S_{D,B}$, using $2n$ bits to indicate it.
With this, Bob can simulate the meta-operations for all sequences $D^*\subseteq[u]$, rejecting if it probes outside of $S_{D,B}$, or if a query returns false.
In fact, the set of all $D^*$ is a subset of $K\cup B$, thus Alice can send only an index into $K\cup B$ (using $\log\binom{n+\mu}{\mu}$ bits) for Bob to learn $D$.
Then, Alice sends $S_{D,A}$ (Bob can derive $S_{D,A}\cap S_{D,B}$)
and the cells that differ between $C_{D,A}(S_{D,A}\cap S_{D,B})$ and $C_{D,B}(S_{D,A}\cap S_{D,B})$.
We use $4n$ bits to indicate the cells and $\mu/\Cdelta$ words for the cell contents.

With this, Bob can recover $C_{D,A}$ and thus learn $(K\setminus D)\cup A$.
Recall that he already knows $D$, and thus the last step for
Alice is to send the index of $A$ in $(K\setminus D)\cup A$ (using $\log\binom{n}{\mu}$ bits).
\begin{align*}
H(M_{\text{send}}|W=1)\leq
\log\binom{u}{\mu}+\log\binom{u}{n}+R+2n+\log\binom{n+\mu}{\mu}+4n+\frac{\mu}{\Cdelta}w+\log\binom{n}{\mu}
\end{align*}
Recall that $\log u=w$ and $u\geq n^2$.
Using $H(M_{\text{send}}|W=1)\geq H(M_{\text{learn}}|W=1)$, we get:
\begin{align*}
-\log \Pr{W=1}
&\geq
\log\binom{u-n-\mu}{\mu/2}
+\brk[s]*{\log\binom{u-n}{\mu} -\log\binom{u}{\mu}}
-\log\binom{n+\mu}{\mu}
-R-6n-\frac{\mu}{\Cdelta}w
\\&\geq
\brk[s]*{\log\binom{u}{\mu/2}-\frac{\mu}{2}}-{\mu}
-\brk[s]*{\mu\log\frac{n+\mu}{\mu}+2\mu}
-R-6n -\frac{\mu}{\Cdelta}\log u
\\&\geq
\frac{\mu}{2}\log\frac{u}{\mu}
-\mu\log\frac{n}{\mu}-5\mu
-R-6n -\frac{\mu}{\Cdelta}\log u
=\mu\log\brk*{\frac{u^{3/8}\cdot \mu^{1/2}}{32n}}-R-6n
\\&
\geq\mu\log\brk*{\frac{u^{3/8}}{32n^{1/2}}\cdot\brk*{\frac{\Cfactor}{\log n}}^{1/2}}-\frac{7}{\Cfactor}\mu\log n
\geq \mu\log\brk*{\frac{u^{9/128}}{4\log^{1/2} n}}\geq \frac{\mu}{16}\log u~.
\end{align*}
where we used the given inequalities
$n\leq \sqrt{u}$ and
$\max\{R,n\}\leq \frac{\mu}{\Cfactor}\log n$
and repeatedly applied
$0\leq\log\binom{u}{n}-n\log\frac{u}{n}\leq 2n$ and
$\log\binom{u}{n}-\log\binom{u'}{n}\leq n\log\frac{u-n}{u'-n}\leq n$  
for $u\geq u'\geq (u+n)/2$.
\end{proof}

\end{appendices}

\end{document}